\documentclass[11pt]{amsart}
\def\isdraft{0}

\usepackage{txfonts}
\usepackage[dvipsnames]{xcolor}
\usepackage{bbm}
\usepackage{graphicx}
\usepackage{bbm}
\usepackage{amsaddr} 
\usepackage{mathtools}
\usepackage{amsmath,amssymb,amsfonts,dsfont}
\usepackage{prettyref} 
\usepackage[utf8]{inputenc}
\usepackage{enumitem}
\usepackage[hyphens]{url}
\usepackage{tikz-cd}
\usepackage{tikz}
\usetikzlibrary{positioning}
\usetikzlibrary{arrows}
\usepackage{bussproofs}
\usepackage[mathscr]{euscript}
\usepackage{hyperref} 
\usepackage{amsthm}
\usepackage{theapa}
\usepackage{a4wide}
\usepackage[color=black,textcolor=white\if\isdraft0,disable\fi]{todonotes}
\usepackage{stackengine}
\usepackage[normalem]{ulem}
\usepackage{stmaryrd}
\usepackage{wrapfig}
\usepackage{pdfpages}
\usepackage{marginnote}
\usepackage{float}
\graphicspath{{fig/}}
\usetikzlibrary{arrows,automata,topaths,matrix,positioning,fit}
\usepgflibrary{shapes.geometric}
\usetikzlibrary{shapes.geometric}
\tikzset{every state/.style={minimum size=0pt}}
\usepackage{footnote}
\usepackage{float}
\usepackage{tabularx}
\usepackage{pifont} 

\newlist{todolist}{itemize}{2}
\setlist[todolist]{label=$\square$}

\newtheorem{theorem}{Theorem}

\newtheorem{corollary}[theorem]{Corollary}
\newtheorem{fact}[theorem]{Fact}

\theoremstyle{definition} 

\newtheorem{definition}[theorem]{Definition}

\newtheorem{example}[theorem]{Example}

\newtheorem{remark}[theorem]{Remark}

\newrefformat{§}{§\ref{#1}}
\newrefformat{algorithm}{Algorithm \ref{#1}}
\newrefformat{a}{Answer \ref{#1}}
\newrefformat{appendix}{Appendix \ref{#1}}
\newrefformat{claim}{Claim \ref{#1}}
\newrefformat{conclusion}{Conclusion \ref{#1}}
\newrefformat{convention}{Convention \ref{#1}}
\newrefformat{conjecture}{Conjecture \ref{#1}}
\newrefformat{C}{Corollary \ref{#1}}
\newrefformat{equ}{(\ref{#1})}
\newrefformat{e}{Example \ref{#1}}
\newrefformat{exe}{Exercise \ref{#1}}
\newrefformat{d}{Definition \ref{#1}}
\newrefformat{done}{Done \ref{#1}}
\newrefformat{f}{Fact \ref{#1}}
\newrefformat{fig}{Figure \ref{#1}}
\newrefformat{h}{Hypothesis \ref{#1}}
\newrefformat{idea}{Idea \ref{#1}}
\newrefformat{i}{Item \ref{#1}}
\newrefformat{l}{Lemma \ref{#1}}
\newrefformat{note}{Note \ref{#1}}
\newrefformat{n}{Notation \ref{#1}}
\newrefformat{o}{Observation \ref{#1}}
\newrefformat{problem}{Problem \ref{#1}}
\newrefformat{p}{Proposition \ref{#1}}
\newrefformat{pseudocode}{Pseudocode \ref{#1}}
\newrefformat{Q}{Q. \ref{#1}}
\newrefformat{r}{Remark \ref{#1}}
\newrefformat{table}{Table \ref{#1}}
\newrefformat{t}{Theorem \ref{#1}}
\newrefformat{Todo}{Todo \ref{#1}}
\newrefformat{w}{Warning \ref{#1}}

\title{Neural logic programs and neural nets}
\author{
	Christian Anti\'c
}
\address{
	christian.antic@icloud.com\\
	Vienna University of Technology\\
	Vienna, Austria
}

\begin{document}
\begin{abstract}
	Neural-symbolic integration aims to combine the connectionist subsymbolic with the logical symbolic approach to artificial intelligence. In this paper, we first define the answer set semantics of (boolean) neural nets and then introduce from first principles a class of neural logic programs and show that nets and programs are equivalent.
\end{abstract}

\maketitle

\section{Introduction}

Artificial neural nets are inspired by the human brain \cite{McCulloch43,Turing48} with numerous applications in artificial intelligence research such as pattern recognition \cite<cf.>{Bishop06}, deep inductive learning through backpropagation \shortcite<cf.>{Goodfellow16}, and game playing with AlphaGo beating the human champion in the game of Go \shortcite{Silver17}, just to name a few. Neural nets are at the core of what is known as the \textit{connectionist}\footnote{\url{https://plato.stanford.edu/entries/connectionism/}} or \textit{subsymbolic} approach to AI. The mathematical subjects behind neural nets are analysis, probability theory, and statistics. 

Logic programming, on the other hand, represents the \textit{symbolic} and \textit{logical} approach to (Good Old-Fashioned) AI \cite{Apt90,Lloyd87}.\footnote{\url{https://en.wikipedia.org/wiki/Symbolic_artificial_intelligence}} It has its roots in mathematical logic and automated reasoning with discrete mathematics, particularly logic, algebra, and combinatorics as its main mathematical tools.

Both worlds, the subsymbolic and the symbolic, have their strengths and weaknesses. Logical formalisms can be interpreted by humans and have a clear formal semantics which is missing for neural nets. Connectionist systems, on the other hand, have a remarkable noise-tolerance and learning capability which is missing for logical formalisms (a notable exception is inductive logic programming \shortcite{Muggleton91}).

Neural-symbolic integration therefore aims at integrating both the symbolic and subsymbolic worlds \cite{Valiant08} \shortcite<cf.>{Garcez02,Garcez09,Garcez15,DeRaedt20}. Compared to the field's short existence, its successes are remarkable and can be found in various fields such as bioinformatics, control engineering, software verification and adaptation, visual intelligence, ontology learning, and computer games \shortcite{Borges11,dePenning11,Hitzler05}. Moreover, it is related to statistical relational learning and probabilistic logic learning \shortcite{DeRaedt08,deRaedt08a,DeRaedt20,Getoor07}. 

The {\bf main contributions} of this paper can be summarized as follows:
\begin{enumerate}
	\item 
		First, in \prettyref{§:Neural_nets} we define an answer set semantics \cite{Gelfond91} \shortcite<cf.>{Baral03,Lifschitz19,Brewka11a,Eiter09} for (boolean) neural nets by defining an immediate consequence and a Fitting operator \cite{vanEmden76,Fitting02} of a net and by applying Approximation Fixed Point Theory (AFT) \shortcite{Denecker00,Denecker04,Denecker12}, which is a non-monotonic generalization of the well-known Knaster-Tarski Theory \cite{Tarski55} of monotonic lattice operators with numerous applications in answer set programming \shortcite<e.g.>{Pelov04,Antic13,Antic20}. To the best of our knowledge, this paper is the first to recognize that such a semantics can be given to (boolean) neural nets (under the assumption that we can remember a neuron's active state). 
	\item 
		Second, in \prettyref{§:NLP} we introduce from first principles neural logic programs as programs assembled from neurons. We define the least model semantics of positive programs, and the answer set semantics of arbitrary programs again by applying AFT; moreover, we define the FLP-answer set semantics in the ordinary sense in terms of the Faber-Leone-Pfeifer reduct \shortcite{Faber04,Faber11}.
	\item 
		Finally, in \prettyref{§:Equivalences} we prove that neural nets and neural logic programs are equivalent.
\end{enumerate}

In a sense, our approach is dual to the well-known \textit{core method} \shortcite{Hoelldobler94,Hoelldobler99} where one starts with a propositional logic program and then constructs a neural net that simulates that program. The core method is the basis for the simulation-based work on neural-symbolic integration as presented in \citeA{Garcez09}. In a broader sense, this paper is a further step towards neural-symbolic integration.

\section{Neural nets}\label{§:Neural_nets}

In this section, we first recall (boolean) neural nets as presented for example in \citeA[§3]{Garcez02} (with the important difference that we assume here that neurons remain activated) where we simplify the functionality of a neuron to a boolean output computed from the weighted sum of its inputs and a given threshold. We then define the least model semantics of positive nets (\prettyref{§:LM_NN}) and the answer set semantics of arbitrary nets (\prettyref{§:ASP_NN}), which appears to be original.

In what follows, we denote the reals by $\mathbb R$, the positive reals by $\mathbb R^+$, $\mathbb R^{-\infty}:=\mathbb R\cup\{-\infty\}$, and the booleans by $\mathbb B=\{0,1\}$. We denote the cardinality of a set $X$ by $|X|$. We define the empty sum $\sum_{i\in\emptyset}a_i:=-\infty$.\footnote{This convention will be used in \prettyref{equ:T_N(0)}.} A \textit{\textbf{prefixed point}} of a function $f:L\to L$ on a lattice $L=(L,\leq)$ is any element $a\in L$ satisfying $f(a)\leq a$, and we call $a$ a \textit{\textbf{fixed point}} of $f$ iff $f(a)=a$.

\subsection{Neurons and nets}\label{§:Neurons_and_sets}

Let $A$ be a set of \textit{\textbf{neurons}} where each neuron $a\in A$ is determined by its \textit{\textbf{threshold}} $\theta(a)\in\mathbb R^{-\infty}$. 

A (\textit{\textbf{neural}}) \textit{\textbf{net}} over $A$ is a finite weighted directed graph $N$ with neurons from $A$ as vertices and edges $a\xrightarrow{w_{ab}}b$, with $w_{ab}\in\mathbb R$. In case $w_{ab}\neq 0$ we say that $a$ and $b$ are \textit{\textbf{connected}} and in case $w_{ab}=0$ we say that $a$ and $b$ are \textit{\textbf{disconnected}}. A net is \textit{\textbf{positive}} iff it contains only positive weights. We write $a\in N$ in case the neuron $a\in A$ appears in $N$. We define the \textit{\textbf{body}} of a neuron $a\in N$ by
\begin{align*} 
	b_N(a):=\{b\in N\mid w_{ba}\neq 0\}.
\end{align*} A \textit{\textbf{fact}} is a neuron with empty body and no weights, and given a fact $a$ we always assume $\theta(a)=-\infty$, and we assume $\theta(a)\neq -\infty$ in case $a$ is not a fact.\footnote{This will allow us to initiate the computation process of a net (see \prettyref{equ:T_N(0)}).} We denote the facts in $N$ by $facts(N)$. The facts will represent the input signals. 

An \textit{\textbf{ordinary net}} is a net $N$ so that $\theta(a)=|b_N(a)|$, for all $a\in N$, and $w_{ba}=1$ for all $b\in b_N(a)$. Intuitively, in an ordinary net $N$ a neuron $a$ ``fires'' iff each of its body neurons in $b_N(a)$ ``fires.''

\subsection{Least model semantics of positive nets}\label{§:LM_NN}

An \textit{\textbf{interpretation}} is any subset of $A$ and we denote the set of all interpretations over $A$ by $\mathbb I_A$. We can interpret each interpretation $I$ as a function $I:A\to\mathbb B$ so that $I(a)=1$ iff $a\in I$.

In the literature \cite<e.g.>[§3]{Garcez09}, the functionality of a neural net is given with respect to a time point $t$ and the activation of a neuron $a$ at $t$ usually means that $a$ is inactive at $t+1$ unless there is a recurrenct connection from $a$ to itself. In this paper, we take a different approach as we assume that once a neuron $a$ is activated it remains active or, in another interpretation, it is remembered that $a$ was active. This allows the net to reach stable configurations which we identify with stable (or answer set) models. This will allow us in \prettyref{§:Equivalences} to show that nets and programs are equivalent.

\begin{definition} Define the (\textit{\textbf{immediate consequence}}) \textit{\textbf{operator}} of $N$, for every interpretation $I$, by
\begin{align}\label{equ:T_N} 
	T_N(I):=\left\{a\in N \;\middle|\; \sum_{b\in b_N(a)}w_{ba}I(b)\geq\theta(a)\right\}.
\end{align} The operator $T_N$ of a \textit{\textbf{positive}} net $N$ ($w_{ba}\geq 0$) is monotone in the sense that
\begin{align*} 
	I\subseteq J \quad\text{implies}\quad T_N(I)\subseteq T_N(J).
\end{align*} We call an interpretation $I$ a \textit{\textbf{model}} of $N$ iff $I$ is a prefixed point of $T_N$, and we call $I$ a \textit{\textbf{supported model}} of $N$ iff $I$ is a fixed point of $T_N$.
\end{definition}

Since the set of all interpretations over $A$ is a complete lattice with respect to union and intersection, it follows by the well-known Knaster-Tarksi Theory \cite{Tarski55} that for a positive net $N$, the operator $T_N$ has a \textit{\textbf{least fixed point}} which can be obtained via a bottom-up iteration of the form
\begin{align*} 
	T_N^0&:=\emptyset\\
	T_N^{n+1}&:=T_N(T_N^n)\\
	T_N^\infty&:=\bigcup_{n\geq 0}T_N^n.
\end{align*} We call $T_N^\infty$ the \textit{\textbf{least model}} of $N$. Notice that this bottom-up computation can only be initiated since we have assumed $\theta(a)=-\infty$ iff $a$ is a fact in $N$, which implies
\begin{align}\label{equ:T_N(0)} 
	T_N(\emptyset)=\left\{a\in N \;\middle|\; \sum_{b\in\emptyset}b=-\infty\geq\theta(a) \right\}=facts(N).
\end{align} This means that a positive net with no facts (i.e. no input signals) has always the empty least model.

The definition of an immediate consequence operator of a neural net and the associated least model of a positive net appears to be new.

\subsection{Answer set semantics of arbitrary nets}\label{§:ASP_NN}

The immediate consequence operator of an arbitrary neural net possible containing negative weights may be non-monotonic which means that its least fixed point may not exist. Approximation Fixed Point Theory (AFT) \shortcite{Denecker00,Denecker04,Denecker12,Fitting02} has been designed exactly for dealing with non-monotonic operators and it can be seen as a generalization of the Knaster-Tarski Theory from monotonic to non-monotonic lattice operators. 

In this section, we use AFT to define the answer set semantics of neural nets, which appears to be original, by following the standard procedure for defining answer sets in terms of the 3-valued Fitting operator (\prettyref{d:AS}).

\begin{definition} A pair of interpretations $(I,J)$ is a \textit{\textbf{3-interpretation}} iff $I\subseteq J$. This can be interpreted as follows:
\begin{itemize}
	\item $a\in I$ means that $a$ is true,
	\item $a\in J-I$ means that $a$ is undefined,
	\item $a\not\in J$ means that $a$ is false.
\end{itemize} 
\end{definition}


\begin{definition} Define the \textit{\textbf{precision ordering}} between 3-interpretations by
\begin{align*} 
	(I,J)\subseteq_p (I',J') \quad\text{iff}\quad I\subseteq I'\subseteq J'\subseteq J.
\end{align*}
\end{definition}


\begin{definition} Define the \textit{\textbf{Fitting operator}} \cite<cf.>{Fitting02} of a net $N$ by
\begin{align*} 
	\Phi_N(I,J):= \left\{a\in N \;\middle|\; \sum_{b\in b_N(a)}w_{ba}K(b)\geq\theta(a),\text{ for all $I\subseteq K\subseteq J$}\right\}.
\end{align*}
\end{definition}

Notice that we have
\begin{align*} 
	\Phi_N(I,I)=T_N(I).
\end{align*} The Fitting operator is monotone with respect to the precision ordering, that is,
\begin{align*} 
	(I,J)\subseteq_p (I',J') \quad\text{implies}\quad \Phi_N(I,J)\subseteq\Phi_N(I',J').
\end{align*} This implies that the operator $\Phi_N(\bullet,I)$ is monotone on the complete lattice of interpretations $\emptyset\subseteq J\subseteq I$ and thus has a least fixed point denoted by $\mathrm{lfp}(\Phi_P(\bullet,I))$. We therefore can define the operator
\begin{align*} 
	\Phi_N^\dagger(I):=\mathrm{lfp}(\Phi_N(\bullet,I)).
\end{align*} 

\begin{definition}\label{d:AS} We call $I$ an \textit{\textbf{answer set}} of $N$ iff $I=\Phi_N^\dagger(I)$.
\end{definition}

The definition of the Fitting operator of a neural net and the associated answer set semantics appears to be new.

\subsection{Acyclic nets}

A net is \textit{\textbf{acyclic}} iff it contains no cycle of non-zero weighted edges. Notice that the neurons in an acyclic net $N$ can be partitioned into \textit{\textbf{layers}} where each neuron only has incoming edges from neurons of lower level. An \textit{\textbf{$n$-layer feed-forward net}} (or \textit{\textbf{$n$-net}}) is an acyclic net $N=N_1\cup\ldots\cup N_n$ (disjoint union) such that $N_i$ contains the neurons of level $i$, $1\leq i\leq n$; we call $N_1$ the \textit{\textbf{input layer}} and $N_n$ the \textit{\textbf{output layer}}. Recall that we assume $\theta(a)=0$ for all input neurons $a\in N_1$ since $b_N(a)$ is empty means that each $a\in N_1$ is a fact (see \prettyref{§:Neurons_and_sets}).

Every $n$-net $N=N_1\cup\ldots\cup N_n$ computes a (boolean) function $f_N:\mathbb I_{N_1}\to\mathbb I_{N_n}$ (notice that $I$ may contain only neurons from the input layer $N_0$ and $f_N(I)$ contains only neurons from the output layer $N_n$) by
\begin{align*} 
	f_N=T_{N_n}\circ\ldots\circ T_{N_1}
\end{align*} so that for each interpretation $I\in\mathbb I_{N_1}$,
\begin{align*} 
	f_N(I)=T_{N_n}(\ldots T_{N_2}(T_{N_1}(I)).
\end{align*}

\section{Neural logic programs}\label{§:NLP}

In this section, we introduce neural logic programs as programs assembled from neurons.

\subsection{Syntax}

Let $A$ be a finite set of neurons. A (\textit{\textbf{neural logic}}) \textit{\textbf{program}} over $A$ is a finite set of (\textit{\textbf{neural}}) \textit{\textbf{rules}} of the form
\begin{align}\label{equ:r} 
	a_0\xleftarrow{\mathbf w}a_1,\ldots,a_k,\quad k\geq 0,
\end{align} where $a_0,\ldots,a_k\in A$ are neurons and $\mathbf w=(w_{a_1a_0},\ldots,w_{a_ka_0})\in\mathbb R^k$, so that $w_{a_ia_0}\neq 0$ for all $1\leq i\leq k$, are weights. A rule of the form \prettyref{equ:r} is \textit{\textbf{positive}} iff all weights $w_1,\ldots,w_k\geq 0$ are positive and a program is \textit{\textbf{positive}} iff it consists only of positive rules. It will be convenient to define, for a rule $r$ of the form \prettyref{equ:r}, the \textit{\textbf{head}} of $r$ by $h(r):=a_0$ and the \textit{\textbf{body}} of $r$ by $b(r):=\{a_1,\ldots,a_k\}$. A program is \textit{\textbf{minimalist}} iff it contains at most one rule for each rule head.

We define the \textit{\textbf{dependency graph}} of $P$ by $dep(P):=(A_P,E_P)$, where $A_P$ are the neurons occurring in $P$, and there is an edge $a\xrightarrow{w_{ba}}b$ in $E_P$ iff there is a rule
\begin{align*} 
	a\xleftarrow{(\ldots,w_{ba},\ldots)}b_1,\ldots,b_{i-1},b,b_{i+1},b_k\in P,\quad k\geq 1.
\end{align*} Notice that the dependency graph of a program is a net! A program is \textit{\textbf{acyclic}} iff its dependency graph is acyclic. Similar to acyclic nets, the neurons $A_P$ occurring in an acyclic program $P$ can be partitioned into \textit{\textbf{layers}} $A_P=A^1_P\cup\ldots\cup A^n_P$ (disjoint union) such that for each rule $r\in P$, if $h(r)\in A^i_P$ then $b\in A^{i-k}_P$, $1\leq k\leq i-1$, for every $b\in b(r)$. An \textit{\textbf{$n$-program}} is an acyclic program which has a partitioning into $n$ layers.

An \textit{\textbf{ordinary rule}} is a rule of the form \prettyref{equ:r} with $\mathbf w=(1,\ldots,1)\in\mathbb R^k$ and $\theta(a_0)=k$ written simply as
\begin{align}\label{equ:r2} 
	a_0\leftarrow a_1,\ldots,a_k.
\end{align} An \textit{\textbf{ordinary program}} consists only of ordinary rules.

\subsection{Answer set semantics}

We now define the semantics of a neural logic program. As for neural nets, an \textit{\textbf{interpretation}} of a program is any subset of $A$. 

\begin{definition} The semantics of ordinary programs is defined as for ordinary propositional Horn logic programs inductively as follows:
\begin{itemize}
	\item for a neuron $a\in A$, $I\models a$ iff $a\in I$;
	\item for a set of neurons $B\subseteq A$, $I\models B$ iff $B\subseteq I$;
	\item for an ordinary rule $r$ of the form \prettyref{equ:r2}, $I\models r$ iff $I\models b(r)$ implies $I\models h(r)$;
	\item for an ordinary program $P$, $I\models P$ iff $I\models r$ for each $r\in P$.
\end{itemize}
\end{definition}


\begin{definition} We define the semantics of neural logic programs inductively as follows:
\begin{itemize}
	\item For a neuron $a\in A$, $I\models a$ iff $a\in I$.
	\item For a rule $r$ of the form \prettyref{equ:r}, we define
	\begin{align*} 
		I\models r \quad\text{iff}\quad \sum_{b\in b(r)}w_{ba}I(b)\geq\theta(h(r)) \text{ implies } I\models h(r).
	\end{align*}
	\item For a neural logic program $P$, $I\models P$ iff $I\models r$ for every $r\in P$, in which case we call $I$ a \textit{\textbf{model}} of $P$.
\end{itemize}
\end{definition}

\begin{definition} Define the (\textit{\textbf{immediate consequence}}) \textit{\textbf{operator}} \shortcite{vanEmden76} of $P$, for every interpretation $I$, by
\begin{align*} 
	T_P(I):=\left\{h(r) \;\middle|\; r\in P, \sum_{b\in b(r)}w_{ba}I(b)\geq\theta(h(r))\right\}.
\end{align*}
\end{definition}

Notice the similarity to the immediate consequence operator of a net in \prettyref{equ:T_N} which will be essential in \prettyref{§:Equivalences}. If $P$ is ordinary, then we get the ordinary immediate consequence operator of \shortciteA{vanEmden76}. 

\begin{example} Consider the program
\begin{align*} 
	P:= \left\{
	\begin{array}{l}
		a\\
		b\xleftarrow{w}a		
	\end{array}
	\right\}.
\end{align*} We have
\begin{align*} 
	a\in T_P(\{a\}) \quad\text{iff}\quad w\{a\}(a)\geq\theta(b) \quad\text{iff}\quad w\geq\theta(b).
\end{align*} So for example if $w=0$ and $\theta(b)>0$, the least model of $P$ is $\{a\}$ which differs from the least model $\{a,b\}$ of the ordinary program which we obtain from $P$ by putting $w=1$ and $\theta(b)=1$.
\end{example}

\begin{fact}\label{f:prefixed} An interpretation $I$ is a model of $P$ iff $I$ is a prefixed point of $T_P$.
\end{fact}

An interpretation $I$ is a \textit{\textbf{supported model}} of $P$ iff $I$ is a fixed point of $T_P$. A program is \textit{\textbf{supported model equivalent}} to a net iff they have the same supported models. As for nets, the operator of a positive program $P$ is monotone and thus has a least fixed point, which by \prettyref{f:prefixed} is a model of $P$, called the \textit{\textbf{least model}} of $P$. A program $P$ is \textit{\textbf{subsumption equivalent}} \cite{Maher88} to a net $N$ iff $T_P=T_N$. Moreover, a positive program is \textit{\textbf{least model equivalent}} to a net iff their least models coincide. As for nets, the operator of a non-positive program may be non-monotonic which means that that its least fixed point may not exists. We can define an answer set semantics of arbitrary programs in the same way as for arbitrary nets by defining the operators $\Phi_P$ and $\Phi_P^\dagger$ and by saying that $I$ is an \textit{\textbf{answer set}} of $P$ iff $I=\Phi_P^\dagger(I)$. Notice that this construction is the standard procedure for defining an answer set semantics of a program in terms of the Fitting operator \shortcite<cf.>{Fitting02,Denecker12}. A program is \textit{\textbf{equivalent}} to a net iff their answer sets coincide.

In contrast to nets, we can define the answer set semantics of a neural logic program in a direct way using the (\textit{\textbf{Faber-Leone-Pfeifer}}) \textit{\textbf{reduct}} \shortcite{Faber04,Faber11} defined, for every program $P$ and interpretation $I$, by
\begin{align*} 
	P^I:=\{r\in P\mid I\models b(r)\}.
\end{align*} We now can say that $I$ is an \textit{\textbf{FLP-answer set}} of $P$ iff $I$ is a $\subseteq$-minimal model of $P^I$. Notice that we cannot define the reduct of a neural net in a reasonable way as we have no notion of ``rule'' for nets, which means that there is no notion of FLP-answer set for nets. 


\section{Equivalences}\label{§:Equivalences}

We are now ready to prove the main results of the paper.

\begin{theorem}\label{t:subsumption} Every neural net is subsumption equivalent to a minimalist neural program.
\end{theorem}
\begin{proof} Given a net $N$, define the minimalist program 
\begin{align*} 
	P_N:=\left\{a\xleftarrow{\mathbf w}b_N(a) \;\middle|\; a\in N,\mathbf w=(w_{ba}\mid b\in b_N(a))\right\}.
\end{align*} Notice that
\begin{align*} 
	N=dep(P_N).
\end{align*} For any interpretation $I$, we have
\begin{align*} 
	T_N(I)
		&=\left\{a\in N \;\middle|\; \sum_{b\in b_N(a)}w_{ba}I(b)\geq\theta(a)\right\}\\
		&=\left\{h(r) \;\middle|\; r\in P_N, \sum_{b\in b(r)}w_{ba}I(b)\geq\theta(h(r))\right\}\\
		&=T_{P_N}(I).
\end{align*}
\end{proof}

\begin{remark} Notice that the neural program $P_N$ constructed in the proof of \prettyref{t:subsumption} is \textit{\textbf{minimalist}} in the sense that it contains at most one rule for each head. Moreover, if $N$ is acyclic then $P_N$ is acyclic as well.
\end{remark}

\begin{corollary} Every neural net is supported model equivalent to a minimalist neural program.
\end{corollary}

\begin{corollary} Every positive neural net is least model equivalent to a positive minimalist neural program.
\end{corollary}

\begin{theorem}\label{t:ordinary} Every ordinary neural net is subsumption equivalent to an ordinary minimalist neural program.
\end{theorem}
\begin{proof} The minimalist program $P_N$ as defined in the proof of \prettyref{t:subsumption} is equivalent to the ordinary program
\begin{align*} 
	\widehat P_N:=\{a\leftarrow b_N(a)\mid a\in N\} 
\end{align*} since
\begin{align*} 
	T_{P_N}(I)
		&= \left\{h(r) \;\middle|\; r\in P_N, \sum_{b\in b_N(a)}I(b)=|b_N(a)|\right\}\\
		&=\left\{h(r) \;\middle|\; r\in \widehat P_N:I\models b(r)\right\}\\
		&=T_{\widehat P_N}(I),
\end{align*} where the first identity follows from the assumption that $N$ is ordinary.
\end{proof}

\begin{theorem} Every neural net is equivalent to a minimalist neural program.
\end{theorem}
\begin{proof} Given a net $N$, for the minimalist program $P_N$ as constructed in the proof of \prettyref{t:subsumption}, we have
\begin{align*} 
	\Phi_N(I,J)
		&=\left\{a\in N \;\middle|\; \sum_{b\in b_N(a)}w_{ba}K(b)\geq\theta(a),\text{ for all $I\subseteq K\subseteq J$}\right\}\\
		&=\left\{r\in P_N \;\middle|\; \sum_{b\in b(r)}w_{bh(r)}K(b)\geq\theta(h(r)),\text{ for all $I\subseteq K\subseteq J$}\right\}\\
		&=\Phi_{P_N}(I,J).
\end{align*}
\end{proof}

\begin{corollary} Every $n$-net is equivalent to an $n$-program.
\end{corollary}
\begin{proof} Given an $n$-net $N$, the program $P_N$ in the proof of \prettyref{t:subsumption} is an $n$-program equivalent to $N$.
\end{proof}

\section{Future work}

In this paper, neurons are boolean in the sense that their outputs are boolean values from $\{0,1\}$, which is essential for the translation into neural logic programs. If we allow neurons to have values in the reals (or any other semiring), then we need to translate them into \textit{weighted} neural logic programs where the semantics is given within the reals (or the semiring) as well \shortcite<cf.>{Droste07,Stueber08,Cohen11}. This is important since learning strategies such as backpropagation work only in the non-boolean setting.

This brings us to the next line of research which is to interpret well-known learning strategies such as backpropagation in the setting of neural logic programs and to analyze the role of the immediate consequence operator in learning.

HEX programs \shortcite{Eiter05} incorporating external atoms into answer set programs can be used to implement neural logic programs where each neuron is implemented by an external atom, given that HEX programs are generalized to allow external atoms to occur in rule heads (which still appears to be missing).

Arguably the most fascinating line of future research is to lift the concepts and results of this paper from propositional to \textit{first-order} neural logic programs. This requires to resolve the fundamental question of what a ``first-order neuron'' is, which is related to the problem of \textit{variable binding} \shortcite<cf.>{Smolensky90,Sun94,Browne99,Feldmann13} at the \textit{core} of neural-symbolic integration \shortcite<cf.>{Bader08}. This line of work will be related to the recently introduced neural stochastic logic programs \shortcite{Manhaeve21,Winters22} which extend probabilistic logic programs with neural predicates.

Finally, it appears fundamental to introduce and study the sequential composition \cite{Antic21-1,Antic21-2} of neural logic programs for program composition and decomposition, since it provides an \textit{algebra of neural logic programs} and thus an \textit{algebra of neural nets} which can be used to define \textit{neural logic program proportions} \cite{Antic23-23} of the form ``$P$ is to $Q$ what $R$ is to $S$'' and thus analogical proportions between neural nets. Even more interesting are ``mixed'' proportions of the form $P:R::M:N$ where $P,R$ are ordinary programs and $M,N$ are neural programs representing neural nets, which ultimately gives us an algebraic connection between pure logic programs and neural nets in the form of proportional functors $F$ satisfying $P:R::F(P):F(R)$ \cite{Antic22-4}.
\todo[inline]{}

\section{Conclusion}

In this paper, we defined the least model semantics of positive and the answer set semantics of arbitrary (boolean) neural nets. After that we defined from first principles the class of neural logic programs and showed that neural nets and neural programs are equivalent. In a broader sense, this paper is a further step towards neural-symbolic integration.

\bibliographystyle{theapa}
\bibliography{/Users/christianantic/Bibdesk/Bibliography,/Users/christianantic/Bibdesk/Publications_J,/Users/christianantic/Bibdesk/Publications_C,/Users/christianantic/Bibdesk/Preprints,/Users/christianantic/Bibdesk/Submitted,/Users/christianantic/Bibdesk/Notes}
\if\isdraft1
\newpage

\section{Ultimate semantics}

We now define the ultimate semantics of a neural net \cite{Denecker04}. For this, define the \textit{\textbf{ultimate 3-valued immediate consequence operator}} of $N$, for every 3-interpretation $(I,J)$, by
\begin{align*} 
	U_N(I,J):= \left(\bigcap_{I\subseteq K\subseteq J}T_N(K),\bigcup_{I\subseteq K\subseteq J}T_N(K)\right).
\end{align*} In analogy to the Fitting operator, the operator $U_N(\bullet,I)$ is monotone on the complete lattice of interpretations $\emptyset\subseteq J\subseteq I$ and thus has a least fixed point denoted by $\mathrm{lfp}(U_P(\bullet,I))$. We therefore can define the operator
\begin{align*} 
	U_N^\dagger(I):=\mathrm{lfp}(U_P(\bullet,I)).
\end{align*} We call $I$ an \textit{\textbf{ultimate answer set}} of $N$ iff $I=U_N^\dagger(I)$.

Two nets $M,N$ are \textit{\textbf{ultimately equivalent}} iff $U_M=U_N$. Notice that subsumption equivalent nets are ultimately equivalent.

\section{Comparison to answer set programs}

In this section, we compare neural logic programs and ordinary \textit{\textbf{answer set programs}} which are finite sets of \textit{\textbf{rules}} of the form
\begin{align} 
	a_0\leftarrow a_1,\ldots,a_k,not\,a_{k+1},\ldots,not\,a_m,\quad m\geq k\geq 0,
\end{align} where $a_0,\ldots,a_m$ are ordinary atoms. 

The question is whether each neural program is equivalent to an answer set program and vice versa...

\section{}

\begin{theorem}\label{t:FLP} Every answer set is an FLP-answer set.
\end{theorem}
\begin{proof} 
\todo[inline]{}
\end{proof}

The following counterexample shows that the converse of \prettyref{t:FLP} fails in general.

\begin{example} 
\todo[inline]{}
\end{example}

\section{Related work}

In neural probabilistic logic programs \cite{todo}, neural networks parametrise probabilistic logic programs \cite{todo}.

Deep probabilistic programming \shortcite{Tran17,Bingham19}... 

DeepStochLog \cite{todo}...

DeepSeaProbLog \cite{todo}... continuous random variables ... 

To the best of our knowledge, none of the mentioned neural-symbolic formalisms use neurons as basic building blocks of programs... 

In a sense, our approach is dual to the well-known \textit{core method} where logic programs are mapped to neural nets... \cite{todo}

\fi
\end{document}